\pdfoutput=1
\documentclass{toc}

\tocdetails{%
volume=13, number=6, year=2017, firstpage=1,
specissue={\cccBG},  %
received={April 30, 2016},
revised={June 7, 2017},
published={September 1, 2017},
doi={10.4086/toc.2017.v013a006},
title={Arithmetic Circuits with Locally Low Algebraic Rank},
author={Mrinal Kumar and Shubhangi Saraf},
plaintextauthor={Mrinal Kumar, Shubhangi Saraf},
acmclassification={F.1.3},
amsclassification={68Q15, 68Q17},
keywords={algebraic rank, arithmetic circuits, hitting sets, lower bounds, non-homogeneous depth-4 circuits, partial derivatives, polynomial identity testing, projected shifted partials}
}

\newtheorem{observation}[theorem]{Observation}

\DeclareMathOperator{\size}{Size}
\DeclareMathOperator{\poly}{Poly}
\DeclareMathOperator{\pr}{Pr}
\DeclareMathOperator{\supp}{Support}

\newcommand{\qspnew}{\Sigma{\Pi^{(k)}}\Sigma\Pi}
\newcommand{\qspnewn}{\Sigma{\Pi^{(n)}}\Sigma\Pi}
\newcommand{\qspnewbounded}{\Sigma{\Pi^{(k)}}\Sigma\Pi^{[d]}}
\newcommand{\qspgeneral}{\Sigma{\Gamma^{(k)}}\Sigma\Pi}

\newcommand\NW{\mathsf{NW}}
\newcommand\gG{\mathsf{NW\circ Lin_{q, n, e, p}}}
\newcommand\gH{\mathsf{NW\circ Lin}}

\renewcommand\dim{\mathsf{Dim}}

\newcommand\h{\mathsf{Hom}}

\newcommand\mult{\mathsf{Mult}}

\newcommand{\VP}{\cclass{VP}}
\newcommand{\VNP}{\cclass{VNP}}
\newcommand{\NP}{\cclass{NP}}
\newcommand{\cP}{\cclass{P}}

\newcommand\F{{\mathbb{F}}}
\newcommand{\field}[1]{\mathbb{#1}}
\mathchardef\mhyphen="2D 

\begin{document}

\begin{frontmatter}
\title{Arithmetic Circuits with\\ Locally Low Algebraic Rank\titlefootnote{A 
conference version of this paper appeared in the 
 Proceedings of the Conference on Computational Complexity, 
 2016~\cite{KS-lowalgebraicrank}.}}

\author[kumar]{Mrinal Kumar\thanks{Research supported in part by NSF grant CCF-1253886 and by a Simons Graduate Fellowship.}}
\author[saraf]{Shubhangi Saraf\thanks{Research supported by NSF grant CCF-1350572.}}

\begin{abstract}
In recent years, there has been a flurry of activity towards proving lower bounds for %
homogeneous depth-4 arithmetic circuits (Gupta et al., Fournier et al., Kayal et al., Kumar-Saraf), which has brought us very close 
 to statements that are known to imply $\cclass{VP} \neq \cclass{VNP}$. It is open if these techniques can go beyond homogeneity, and in this paper we make progress 
in this direction  %
by considering depth-4 circuits of \emph{low algebraic rank}, which 
are a natural extension of homogeneous depth-4 arithmetic circuits. 

A  depth-4 circuit is a representation of an $N$-variate, 
degree-$n$   %
polynomial $P$  as 
\[
P = \sum_{i = 1}^T  Q_{i1}\cdot  Q_{i2}\cdot  \cdots Q_{it} \; ,
\]
where the $Q_{ij}$ are given by their monomial expansion. Homogeneity adds the constraint that for every $i \in [T]$, $\sum_{j} \deg(Q_{ij}) = n$. We study an extension, %
 where, for every $i \in [T]$, the \emph{algebraic rank} of the set
$\{Q_{i1},  Q_{i2}, \ldots ,Q_{it}\}$ of polynomials %
is at most some parameter  $k$.  We call this the class of $\qspnew$ circuits. Already for $k = n$, these circuits are a strong generalization of the class of homogeneous depth-4 circuits, where in particular $t \leq n$ (and hence $k \leq n$).

We study lower bounds and polynomial identity tests for such circuits and prove the following results. 

\begin{enumerate}
\item {\bf Lower bounds.}  %
We give an explicit family of polynomials $\{P_n\}$ of
degree~$n$
in $N = n^{O(1)}$ variables in $\VNP$, such that any $\qspnewn$
circuit computing $P_n$ has size at least $\exp{(\Omega(\sqrt{n}\log N))}$.
This strengthens and unifies two lines of work: it generalizes the recent exponential lower bounds for \emph{homogeneous} depth-4 circuits~(Kayal et al.\ and Kumar-Saraf) as well as the Jacobian based lower bounds of  Agrawal et al.\ which worked for $\qspnew$ circuits in the restricted setting where  $T\cdot k \leq n$.  

\item {\bf Hitting sets.}  %
  Let $\qspnewbounded$ be the class of $\qspnew$ circuits with bottom fan-in at most $d$. We show that if $d$ and $k$ are at most $\poly(\log N)$, then there is an explicit hitting set for $\qspnewbounded$ circuits of size quasipolynomially bounded in $N$ and the size of the circuit. This strengthens a result of Forbes
who constructed  %
such
quasipolynomial-size   %
hitting sets in the setting where $d$ and $t$ are at most  $\poly(\log N)$. 
\end{enumerate}

A key technical ingredient of the proofs is a result which states that over any field of characteristic zero (or sufficiently large characteristic),
up to  %
a translation, every polynomial in a set of
polynomials can be written as a function of the polynomials
in a transcendence basis of the set.
We believe this may be of independent interest. We combine this with
methods based on shifted partial derivatives %
to obtain our final results.
\end{abstract}

\iffalse %
%
%
%
%
%
%
%
%
%
%
%
%
%
%
%
%
%
%
%
%
%

\tocarxivcategory{cs.CC}

\fi %

\end{frontmatter}

\section{Introduction}

Arithmetic circuits are natural algebraic analogues of Boolean circuits, with the logical operations being replaced by sum and product operations over the underlying field. Valiant~\cite{Valiant79} developed the complexity theory for algebraic computation via arithmetic circuits and defined the complexity classes $\VP$ and $\VNP$ as the algebraic analogs of complexity classes $\cP$ and $\NP$ respectively. We refer the interested reader to the survey by Shpilka and Yehudayoff~\cite{SY10} for more on arithmetic circuits. 

Two of the most fundamental questions in the study of algebraic computation are the questions of \emph{polynomial identity testing}(PIT)\footnote{Given an arithmetic circuit, the problem is to decide if it computes the identically zero polynomial. In the whitebox
setting  %
we are allowed to look inside the wirings of the circuit, while in the blackbox setting, we can only query the circuit at some points.} and the question of proving \emph{lower bounds} for explicit polynomials. It was shown by structural results known as \emph{depth reductions}~\cite{AV08,  Koiran-depthreduction, Tavenas13} that strong enough lower bounds or PIT results for just (homogeneous) depth-4 circuits, would lead to superpolynomial lower bounds and derandomized PIT for general circuits too. Consequently, depth-4 arithmetic circuits have been the focus of much investigation in the last few years.   

Just in the last few years, we have seen rapid progress in proving lower bounds for homogeneous depth-4 arithmetic circuits, starting with the work of Gupta et al.~\cite{GKKS12} who proved exponential lower bounds for homogeneous depth-4 circuits with bounded bottom fan-in and terminating with the results of Kayal et al.~\cite{KLSS14} and of the authors of this paper~\cite{KS-full}, which showed exponential lower bounds for general homogeneous depth-4 circuits. Any asymptotic improvement in the exponent of these lower bounds would lead to superpolynomial lower bounds for general arithmetic circuits.\footnote{We refer the interested reader to the surveys of recent lower bounds results by Saptharishi~\cite{Saptharishi-survey1, Saptharishi-survey2}.} Most of this progress was based on an understanding of the complexity measure of the family of \emph{shifted partial derivatives} of a polynomial (this measure was introduced by Kayal~\cite{Kayal12}), and other closely related measures.

Although we now know how to use these measure to prove such strong lower bounds for homogeneous depth 4 circuits, the best known lower bounds for non-homogeneous depth three circuits over fields of characteristic zero are just cubic~\cite{SW01, shp01, KST16}, and those for non-homogeneous depth-4 circuits over any field except $\F_2$ are just about superlinear~\cite{Raz10b}. It remains an extremely interesting question to get improved lower bounds for these circuit classes.

In sharp contrast to this state of knowledge on lower bounds, the problem of polynomial identity testing is very poorly understood even for depth three circuits. Till a few years ago, almost all the PIT algorithms known were for extremely restricted classes of circuits and were based on diverse proof techniques (for instance, \cite{DS06, KayalSaxena07, karninshpilka08, KayalSaraf09, KMSV10, SaxenaSeshadhri10, SaxenaSeshadhri10b,  SarafV11, ASSS12, ForbesS13,
OSV14}). The
paper by
Agrawal et al.~\cite{ASSS12} gave a unified proof of several of them.

It is a big question to go \emph{beyond homogeneity} (especially for proving lower bounds) and in this paper we make progress towards this question by considering depth-4 circuits of \emph{low algebraic rank},\footnote{The algebraic rank of a set of polynomials is the size of the maximal subset of this set which are algebraically independent. See \expref{Section}{sec:prelims} for formal definitions.} which are a natural extension of homogeneous depth-4 arithmetic circuits.

A  depth-4 circuit is a representation of an $N$-variate, degree-$n$ polynomial $P$  as 
\[
P = \sum_{i = 1}^T  Q_{i1}\cdot  Q_{i2}\cdot  \cdots Q_{it}
\]
where the $Q_{ij}$ are given by their monomial expansion. Homogeneity adds the constraint that for every $i \in [T]$, $\sum_{j} \deg(Q_{ij}) = n$. 
We study an extension where, for every $i \in [T]$, the \emph{algebraic rank} of the set
$\{Q_{i1},  Q_{i2}, \ldots ,Q_{it}\}$ of polynomials %
is at most some parameter  $k$.  We call this the class of $\qspnew$ circuits. Already for $k = n$, these circuits are a strong generalization of the class of homogeneous depth-4 circuits, where in particular $t \leq n$ (and hence $k \leq n$).

We prove exponential lower bounds for $\qspnew$ circuits for $k \leq n$ and give quasipolynomial time deterministic  polynomial identity tests for $\qspnew$ circuits when $k$ and the bottom fan-in are bounded by $\poly(\log N)$. 
All our results actually hold for a more general class of circuits, where the product gates at the second level can be replaced by an arbitrary circuits whose inputs are polynomials of algebraic rank at most $k$. 
In particular, our results hold for representations of a polynomial $P$ as 
\[
P = \sum_{i = 1}^T  C_i\left(Q_{i1}, Q_{i2}, \ldots, Q_{it}\right)
\]
where, for every $i \in [T]$, $C_i$ is an arbitrary polynomial function of $t$ inputs, and the algebraic rank of the set
$\{Q_{i1},  Q_{i2}, \ldots ,Q_{it}\}$ of polynomials %
is at most some parameter  $k$.

\subsection{Some background and motivation}
Before we more formally define the model and state our results, we give some background and motivation for studying this class of circuits. 
\paragraph{Strengthening of the model of homogeneous depth-4 circuits.}
As already mentioned, we know very strong exponential lower bounds for homogeneous depth-4 arithmetic circuits. In contrast, for general (non-homogeneous) depth-4 circuits, we know only barely superlinear lower bounds, and it is a challenge to obtain improved bounds. $\qspnew$ circuits with $k$ as large as $n$ (the degree of the polynomial being computed), which is the class we study in this paper, is already a significant strengthening of the model of homogeneous depth-4 circuits (since the intermediate degrees could be exponentially large). We provide exponential lower bounds for this model. Note that when $k = N$, $\qspnew$ circuits would capture general depth-4 arithmetic circuits. 

\paragraph{Low algebraic rank and lower bounds.}
In a recent
paper,  %
Agrawal et al.~\cite{ASSS12} studied the notion of circuits of low algebraic rank and by using the Jacobian to capture the notion of algebraic independence, they were able to
prove %
exponential lower bounds for a certain class of arithmetic circuits.\footnote{Even more significantly they also give efficient PIT algorithms for the same class of circuits.} They showed that over fields of characteristic zero, for any set  $\{Q_1, Q_2, \ldots, Q_t\}$ of polynomials of sparsity at most $s$ and algebraic rank $k$, any arithmetic circuit of the form $C(Q_1, Q_2, \ldots, Q_t)$  which computes the determinant polynomial for an $n \times n$ symbolic matrix must
have
$s \geq \exp{(n/k)}$.
Note that  %
if $k = \Omega(n)$, then the lower bound becomes trivial. The lower bounds in this paper strengthen
these results %
in two ways. 

\begin{enumerate}
\item Our lower bounds hold for a
(potentially) richer class of circuits. In the model considered by~\cite{ASSS12}, one imposes a global upper bound $k$ on the rank of
all the $Q_i$   %
feeding into some polynomial $C$. In our model, we can take exponentially many different
sets of polynomials $Q_i$,
each with bounded rank, and apply some polynomial function to each of them and then take a sum.
\item Our lower bounds are stronger---we obtain exponential lower bounds even when $k$ is as large as the degree of the polynomial being computed. 
\end{enumerate}

\paragraph{Algebraic rank and going beyond homogeneity.}
Even though we know exponential lower bounds for homogeneous\footnote{These results,
in fact,
hold for depth-4 circuits with not-too-large formal degree.} depth-4 circuits, the best known lower bounds for non-homogeneous depth-4 circuits are  
barely superlinear~\cite{Raz10b}. 

Grigoriev-Karpinski~\cite{GK98}, Grigoriev-Razborov~\cite{GR00} and Shpilka-Wigderson~\cite{SW01} outlined a program based on ``rank'' to prove lower bounds for arithmetic circuits. They used the notion of ``linear rank'' and used it to prove lower bounds for depth-3 arithmetic circuits in
the following way.  %
Let $C = \sum_{i = 1}^T\prod_{j = 1}^t L_{ij}$ be a depth three (possibly nonhomogeneous) circuit computing a polynomial $P$ of degree-$n$. Now, partition the inputs to the top sum gate to two halves, $C_1$ and $C_2$ based on the rank of the inputs feeding into it in the following way. For each $i \in [T]$, if the linear rank of the set  $\{L_{ij} : j \in [t]\}$ of polynomials is at most $k$ (for some threshold $k$), then include the gate $i$ into the sum $C_1$, else include it into $C_2$. Therefore, $$C = C_1 + C_2\,.$$
Their program had two steps.
\begin{enumerate}
  \item Show that the subcircuit $C_1$ is \emph{weak} with respect to some complexity measure, and thus
prove  %
a lower bound for $C_1$ (and hence $C$) when $C_2$ is trivial. 
\item Also since $C_2$ is ``high rank,'' show that there are many inputs for which $C_2$ is identically zero. Then try to look at restrictions over which $C_2$ is identically zero, and show that the lower bounds for $C_1$ continue to hold.
\end{enumerate}

The following is the natural generalization of this approach to proving lower bounds for depth-4 circuits. 
Let $C = \sum_{i = 1}^T\prod_{j = 1}^t Q_{ij}$ be a depth-4 circuit computing a polynomial $P$ of degree-$n$. 
Note that in general, the formal degree of $C$ could be much larger than $n$. Now, we partition the inputs to the top sum gate to two halves, $C_1$ and $C_2$ based on the \emph{algebraic rank} of the inputs feeding into it in the following way. For each $i \in [T]$, if the algebraic rank of the set  $\{Q_{ij} : j \in [t]\}$ of polynomials is at most $k$ (for some threshold $k$), then we include the gate $i$ into the sum $C_1$ else we include it into $C_2$. Therefore, $$C = C_1 + C_2\,.$$
To implement the G-K, G-R and S-W program, as a first step one would show that the subcircuit $C_1$ is \emph{weak} with respect to some complexity measure,
and thus
prove  %
a lower bound for $C_1$ (and hence $C$) when $C_2$ is trivial. The second step would be to try to look at restrictions over which $C_2$ is identically zero, and show that the lower bounds for $C_1$ continue to hold.

For the case of depth-4 circuits, even the first step of
proving  %
lower bounds when $C_2$ is trivial was not known prior to this work (even for $k = 2$). Our results in this paper are an implementation of this first step, as we
prove   %
exponential lower bounds when the algebraic rank of inputs into each of the product gates is at most $n$ (the degree of the polynomial being computed).

\paragraph{Connections to divisibility testing.}
Recently, Forbes~\cite{Forbes-personal} showed that given two sparse multivariate polynomials $P$ and $Q$, the question of deciding if $P$ divides $Q$ can be reduced to the question of polynomial identity testing for $\Sigma\Pi^{(2)}\Sigma\Pi$ circuits. This question was one of the original motivations for this paper. Although we are unable to answer this question in general,  we make some progress towards it by giving a quasipolynomial identity tests for $\qspnew$ circuits when the various $Q_{ij}$ feeding into the circuit have degree bounded by $\poly(\log N)$ (and we are also able to handle $k$ as large as $\poly(\log N)$).

\paragraph{Low algebraic rank and PIT.} 
Two very interesting PIT results which are also very relevant to the results in this paper are those of Beecken et al.~\cite{BMS11} and those of Agrawal et al.~\cite{ASSS12}. The key idea explored in both these papers is that of algebraic independence. Together, they imply efficient deterministic PIT for polynomials which can be expressed in the form $C(Q_1, Q_2, \ldots, Q_t)$, where $C$ is a circuit of polynomial degree  and $Q_i's$ are either sparse polynomials or product of linear forms, such that the algebraic rank of $\{Q_1, Q_2, \ldots, Q_t\}$ is bounded.\footnote{See \expref{Section}{sec:prelims} for definitions.} This approach was extremely powerful as Agrawal et al.~\cite{ASSS12}
demonstrate   %
that they can use this approach to recover many of the known PIT results, which otherwise had very different proofs techniques. The PIT results of this paper hold for a variation of the model just described and we describe it in more detail in \expref{Section}{section:pitresults}.

\paragraph{Polynomials with low algebraic rank.}
In addition to potential applications to arithmetic circuit complexity, it seems an interesting mathematical question to understand the structure of a set of algebraically dependent polynomials.  In general, our understanding of algebraic dependence is not as clear as our understanding of linear dependence. For instance, we know that if a set of polynomials is linearly dependent, then every polynomial in the set can be written as a linear combination of the polynomials in the basis. However, for higher degree dependencies (linear dependence is dependency of degree-$1$), we do not know any such clean statement. As a significant core of our proofs, we prove a statement of this 
flavor in \expref{Lemma}{lem:using algebraic dependence-intro}.

We now formally define the model of computation studied in this paper, and then state and discuss our results. 
\subsection{Model of computation}
We start with the definition of algebraic dependence. See \expref{Section}{sec:prelims} for more details. 
\begin{definition}[Algebraic independence and algebraic rank]\label{def:alg-indepence}
  Let $\F$ be any field. A set
  \[
    {\cal Q} = \{Q_1, Q_2, \ldots, Q_t\} \subseteq \F[X_1, X_2, \ldots, X_N]
  \]
  of polynomials is said to be algebraically independent over $\F$ if there is no nonzero polynomial $R \in \F[Y_1, Y_2, \ldots, Y_t]$ such that $R(Q_1, Q_2, \ldots, Q_t)$ is identically zero. 

A maximal subset of $\cal Q$ which is algebraically independent is said to be a transcendence basis of $\cal Q$ and the size of such a set is said to be the algebraic rank of $\cal Q$.
\end{definition}

It is known that algebraic independence satisfies the Matroid
property~\cite{Oxley06}, and therefore
the algebraic rank %
is well defined. We are now ready to define the model of computation. 

\begin{definition}~\label{def:lb-model}
Let $\F$ be any field. A $\qspnew$ circuit $C$ in $N$ variables over
$\F$ is a representation of an
$N$-variate %
polynomial as 
\[C =  \sum_{i = 1}^T  Q_{i1}\cdot  Q_{i2}\cdots Q_{it} 
\]  
for some $t, T$ %
such that for each $i \in [T]$, the algebraic rank of the set
$\{Q_{ij} : j \in [t]\}$ of polynomials %
is at most $k$. Additionally, if for every $i \in [T]$ and $j \in [t]$, the degree of $Q_{ij}$ is at most $d$, we say that $C$ is a $\qspnewbounded$ circuit. 
\end{definition}

We will state all our results for $\qspnew$ and $\qspnewbounded$ circuits. However, the results in this paper hold for a more general class of circuits where the product gates at the second level can be replaced by
arbitrary polynomials. This larger class of circuits will be crucially used in our proofs and we
define it formally below.  %
\begin{definition}~\label{def:lb-modelnew}
  Let $\F$ be any field. A $\qspgeneral$ circuit $C$ in $N$ variables over $\F$ is a representation of an
$N$-variate %
polynomial %
as 
\[C =  \sum_{i = 1}^T  \Gamma_i(Q_{i1}, Q_{i2},  \ldots, Q_{it}) \]  
for some $t, T$ %
such that $\Gamma_i$ is an arbitrary polynomial in $t$ variables, and for each $i \in [T]$, the algebraic rank of the set
$\{Q_{ij} : j \in [t]\}$ of polynomials %
is at most $k$. Additionally, if for every $i \in [T]$ and $j \in [t]$, the degree of $Q_{ij}$ is at most $d$, we say that $C$ is a $\qspgeneral^{[d]}$ circuit. 
\end{definition}
\begin{definition}[Size of a circuit]
  The \emph{size} of  a $\qspnew$ or a $\qspgeneral$ circuit $C$ is defined as the maximum of $T$ and the number of monomials in the set  \[
    \left(\bigcup_{i \in [T], j \in [t]}\supp(Q_{ij})\right)\,.
  \]
  Here for a polynomial $Q$, $\supp(Q)$ is the set of all monomials which appear with a non-zero coefficient in $Q$. 
\end{definition}

A $\qspnew$ circuit $C$ for which the polynomials $\{Q_{ij} : i \in [T], j \in [t]\}$ are homogeneous polynomials such that for every $i \in [T]$, \[\sum_{j \in [t]} \deg(Q_{ij}) = \deg(P)\] (where $P$ is the polynomial being computed)\footnote{Observe that in this case, $k \leq t \leq \deg(P)$.} is the class of homogeneous depth-4 circuits. 
If we drop the condition of homogeneity, then in general the value of $t$ could be much larger than $\deg(P)$
and %
the degrees of the $Q_{ij}$ could be much larger than $\deg(P)$. 
Thus, the class of $\qspnew$ circuits with $k$ equaling the degree of the polynomial being computed could potentially be a larger class of circuits compared to that of homogeneous depth-4 circuits. 

Also note that in the definition of $\qspnew$ circuits, the bound on the algebraic rank is local for each $i \in [T]$, and in general, the algebraic rank of the entire set $\{Q_{ij} : i \in [T], j \in [t]\}$ can be as large as $N$.

\subsection{Our results}
We now state our results and discuss how they relate to other known results.

\subsubsection{Lower bounds}
As our first  result, we
give  %
exponential lower bounds on the size of $\qspnew$ circuits computing an explicit polynomial when the algebraic rank ($k$)  is at most the degree ($n$) of the polynomial being computed. 
\begin{theorem}~\label{thm:lower bound}
Let $\F$ be any field of characteristic zero.\footnote{Sufficiently large characteristic suffices.} There exists a family $\{P_n\}$ of polynomials in $\VNP$, such that $P_n$ is a polynomial of degree-$n$ in $N = n^{O(1)}$ variables with $0,1$ coefficients, and for any $\qspnew$ circuit $C$, if $k \leq n$ and if $C$ computes $P_n$ over $\F$, then $$ \size(C) \geq N^{\Omega(\sqrt{n})}\,.$$
\end{theorem} 

\begin{remark}
{}From our proofs it follows that our lower bounds hold for the more general class of $\qspgeneral$ circuits, but for the sake of simplicity, we state our results in terms of $\qspnew$ circuits. We believe it is likely that the lower bounds also hold for a polynomial in $\VP$ and it would be interesting to know if this is indeed true.\footnote{More on this in \expref{Section}{sec:open questions}.}
\end{remark}

\begin{remark}
Even though we state \expref{Theorem}{thm:lower bound} for $k \leq n$, the proof goes through as long as $k$ is any polynomial in $n$ and $N$ is chosen to be an appropriately large polynomial in $n$. 
\end{remark}
\subsubsection{Comparison to known results}\label{sec:known results lb}
As we alluded to in the introduction, $\qspnew$ circuits for $k \geq n$ subsume the class of homogeneous depth-4 circuits. Therefore, 
 \expref{Theorem}{thm:lower bound} subsumes the lower bounds for homogeneous depth-4 circuits~\cite{KLSS14, KS-full} for sufficiently large characteristic. Moreover, it also subsumes and generalizes the lower bounds of Agrawal et al.~\cite{ASSS12} since their lower bounds hold only if the algebraic rank of the entire set  $\{Q_{ij} : i \in [T], j \in [t]\}$ of polynomials is bounded, while for \expref{Theorem}{thm:lower bound}, we only need upper  bounds on the algebraic rank separately for every $i \in [T]$.

\subsubsection{Polynomial identity tests}\label{section:pitresults}

We show that there is a quasipolynomial size hitting set for all polynomials $P \in \qspnew^{[d]}$ for \emph{bounded} $d$ and $k$. More formally, we prove the following theorem.
\begin{theorem}~\label{thm:PIT}
Let $\F$ be any field of characteristic zero.\footnote{Sufficiently large characteristic suffices.}  Then, for every $N$, there exists  a set ${\cal H} \subseteq \F^N$ such that 
\[
\left|{\cal H} \right| \leq \exp(O(\log^{O(1)} N))
\]
and for every nonzero $N$-variate polynomial $P$ over $\F$ which is computable by a $\qspnewbounded$ circuit with $d, k \leq \log N$ and size $\poly(N)$,   there exists an $h \in {\cal H}$ such that $P(h) \neq 0$. Moreover, the set ${\cal H}$ can be explicitly constructed in time  \[
\exp(O(\log^{O(1)} N))\,.
\]
\end{theorem} 
We now mention some  remarks about \expref{Theorem}{thm:PIT}. 

\begin{remark}
It follows from our proof that the hitting set works for the more general class of $\qspgeneral^{[d]}$ circuits with $d, k \leq \log N$, size $\poly(N)$ and formal degree at most $\poly(N)$. 
\end{remark}

\subsubsection{Comparison to known results}\label{sec:known results pit}
The two known results closest to our PIT result are the results of Forbes~\cite{Forbes-personal} and the results of Agrawal et al.~\cite{ASSS12}. Forbes~\cite{Forbes-personal} studies PIT for the case where the number of \emph{distinct inputs} to the second level product gates in a depth-4 circuit with bounded bottom fan-in also bounded (which naturally also bounds the algebraic rank of the inputs), and
constructs  %
quasipolynomial-size
hitting sets for this case. On the other hand, we handle the case where there is no restriction on the number of distinct inputs feeding into the second level product gates, but we need to bound the bottom fan-in as well as the algebraic rank. In this sense, the results in this paper are a generalization of the results of Forbes~\cite{Forbes-personal}. 

Agrawal et al.~\cite{ASSS12} give a construction of
polynomial-size
hitting sets in the case when the total algebraic rank of the set $\{Q_{ij} : i \in [T], j \in [t]\}$ is bounded, but they can work with unbounded $d$. On the other hand, the size of our hitting set  depends exponentially on $d$, but requires only local algebraic dependencies for every $i \in [T]$. So, these two results are not comparable, although there are similarities in the sense that both of them aim to use the algebraic dependencies in the circuit. In general, summation is a tricky operation with respect to designing PIT algorithms (as opposed to multiplication), so it is not clear if the ideas in the work of Agrawal et al.~\cite{ASSS12} can be somehow adapted to prove \expref{Theorem}{thm:PIT}.  

\subsubsection{From algebraic dependence to functional dependence}
Our lower bounds and PIT results crucially use the following lemma, which (informally) shows that over fields of characteristic zero,
up to  %
a translation, every polynomial in a set of
polynomials can be written as a \emph{function} of the polynomials in \emph{transcendence basis}.\footnote{A transcendence basis of a set of polynomials is a maximal subset of the polynomials with the property that its elements are algebraically independent. For more on this see \expref{Section}{sec:prelims}.} We now state the lemma precisely. 
\begin{lemma}[Algebraic dependence to functional dependence]~\label{lem:using algebraic dependence-intro}
  Let $\F$ be any  field of characteristic zero or sufficiently large
positive characteristic.  %
Let ${\cal Q} = \{Q_1, Q_2, \ldots, Q_t\}$ be a set of polynomials in $N$ variables such that the algebraic rank of ${\cal Q}$ equals $k$. Let $d_i=\deg(Q_i)$ ($i\in [t]$) and let  ${\cal B} = \{Q_1, Q_2, \ldots, Q_k\}$ be a maximal algebraically independent subset of ${\cal Q}$. Then, there exists an $\overline{a} = (a_1, a_2, \ldots, a_N)$ in $\F^N$ and polynomials $F_{k+1}, F_{k+2}, \ldots, F_{t}$  in $k$ variables such that $\forall i \in \{k+1, k+2, \ldots, t\}$
$$Q_i(\overline{X} + \overline{a}) = \h^{\leq d_i}\left[F_i(Q_1(\overline{X} + \overline{a}), Q_2(\overline{X} + \overline{a}), \ldots, Q_k(\overline{X} + \overline{a})) \right]\,.$$
Here, for any  polynomial $P$, we use $\h^{\leq i}[P]$
to refer to the sum of homogeneous components of $P$ of degree at most $i$.\footnote{For a more precise definition see \expref{Definition}{def:homog components}.}
\end{lemma}
Even though the lemma seems a very basic statement about the structure of algebraically dependent polynomials, to the best of our knowledge this was not known before. The proof builds upon a result on the structure of roots of multivariate polynomials by Dvir et al.~\cite{DSY09}. Observe that for  \emph{linear} dependence, the statement analogous to that of \expref{Lemma}{lem:using algebraic dependence-intro} is trivially true. We believe that this lemma might be of independent interest (in addition to its applications in this paper). 

In fact, the lemma holds for a random choice of the vector $\overline{a}$ chosen uniformly from a large enough grid  in $\F^N$. 
\begin{remark}\label{rmk:subsequent work}
In a recent result, Pandey et al.~\cite{PSS16} show that this connection between algebraic dependence and functional dependence continues to hold over fields of small characteristic. Consequently, they show that the results of this paper also hold over fields of small characteristic.
\end{remark}

\subsection{Proof overview}

Even though the results in this paper seem related to the results in~\cite{ASSS12} (both exploiting some notion of low algebraic rank), the proof strategy and the way algebraic rank is used are quite different. We now briefly outline our proof strategy.

We first discuss the overview of proof for our lower bound. 

Let $P_n$ be the degree-$n$ polynomial we want to compute, and let $C$ be a $\qspnew$ circuit computing it,  with $k = n$.  Then $C$ can be represented as 
$$C = \sum_{i = 1}^T \prod_{j=1}^t Q_{ij}\,.$$
{}From definitions, we know that for every $i \in [T]$, 
the algebraic rank of the set
$\{Q_{i1},  Q_{i2}, \ldots ,Q_{it}\}$ of polynomials %
is at most $k (=n)$. 
We want to
give  %
a lower bound on the size of $C$.

Instead of proving our result directly for $\qspnew$ circuits, it will be very useful for us to go to the significantly strengthened class of $\qspgeneral$ circuits and prove our result for that class. Thus we think of our circuit $C$ as being expressed as
$$C = \sum_{i = 1}^T  C_{i}(Q_{i1}, Q_{i2}, \ldots, Q_{it})$$  where the $C_i$ can be arbitrary polynomial functions of the inputs feeding into them.
Note that we define the size of a $\qspgeneral$ circuit to be the maximum of the top fan-in $T$, and the maximum of the number of monomials in any of the polynomials $Q_{ij}$ feeding into the circuit. Thus we completely disregard the complexities of the various polynomial function gates at the second level. If we are able to prove a lower bound for this notion of size, then if the original circuit is actually a $\qspnew$ circuit then it will also be as good a lower bound for the usual notion of size. 

Our lower bound has two key steps. In the first step we prove the result in the special case where $t \leq n^2$. In the second step we show how to ``almost'' reduce to the case of $t\leq n^2$. 

\paragraph{Step (1) : $t \leq n^2$. }

In the representation of $C$ as a $\qspgeneral$ circuit, the value of $t$ is at most $n^2$. Lower bounds for this case turn out to be similar to lower bounds for homogeneous depth-4 circuits. In this case we borrow ideas from prior works~\cite{GKKS12, KLSS14, KS-full} and show that the \emph{dimension of projected shifted partial derivatives of $C$} is not too large. Most importantly, we can use the chain rule for partial derivatives to obtain good bounds for this complexity measure, independent of the complexity of the various $C_i$.

Recall however that in our final result, $t$ can be actually much larger than $n^2$. Indeed the circuit $C$ can be very far from being homogeneous, and for general depth-4 circuits,  we do not know good upper bounds on the complexity of shifted partial derivatives or projected shifted partial derivatives. Also, in general, it is not clear if these measures are really small for general depth-4 circuits.\footnote{Indeed, as an earlier result of the authors~\cite{KS-formula} shows, even homogeneous depth-4 circuits can have very large shifted partial derivative complexity.}
It is here that the low algebraic rank of  $\{Q_{i1}, Q_{i2}, \ldots, Q_{it}\}$ proves to be useful, and that brings us to the crux of our argument.

\paragraph{Step (2) : Reducing to the case where $t \leq n^2$. } 

A key component of our proof, which is formalized in \expref{Lemma}{lem:expressing as functions of the basis} shows that over any field of characteristic zero (or sufficiently large characteristic),
up to  %
a translation, every polynomial in a set of
polynomials can be written as a function of the homogeneous components of the polynomials in the transcendence basis. 

More formally, there exists an $\overline{a} \in \F^{N}$ such that $C(\overline{X} + \overline{a})$ can be expressed as  
\[
C(\overline{X} + \overline{a}) = \sum_{i = 1}^T  C_i'(\h[Q_{i1}(\overline{X} + \overline{a})], \h[Q_{i2}(\overline{X} + \overline{a})], \ldots, \h[Q_{ik}(\overline{X} + \overline{a})])
\]
where for a degree-$d$ polynomial $F$, $\h[F]$ denotes the $d+1$-tuple of homogeneous components of $F$. Moreover, $Q_{i1}, Q_{i2}, \ldots, Q_{ik}$ are the polynomials in the transcendence basis. 

The crucial gain in the above transformation is that  the arity of each of the polynomials  $C_i'$ is $(d+1) \times k$ and not $t$ (where $d$ is an upper bound on the degrees of the $Q_{ij}$). Now by assumption $k \leq n$, and moreover without loss of generality we can assume $d \leq n$ since homogeneous components of $Q_{ij}$ of degree larger than $n$ can be dropped since they do not contribute to the computation of a degree-$n$ polynomial. Thus we have essentially reduced to the case where $t \leq n^2$. 

One loss by this transformation is that the polynomials $\{C_i'\}$ might be much more complex and with much higher degrees than the original polynomials $\{C_i\}$. However this will not affect the computation of our complexity measure. Another loss is that we have to deal with the translated polynomial $C(\overline{X} + \overline{a})$. This introduces some subtleties into our computation as it could be that $Q_{ij}(\overline{X})$ is a sparse polynomial but $Q_{ij}(\overline{X}+ \overline{a})$ is far from being sparse. Neither of these issues is very difficult to deal with, and we are able to get strong bounds for the
measure, based on projected shifted partial derivatives, %
for such circuits. The proof of \expref{Lemma}{lem:expressing as functions of the basis} essentially follows from \expref{Lemma}{lem:using algebraic dependence-intro}.

The proof of \expref{Lemma}{lem:using algebraic dependence-intro}  crucially uses a result of Dvir, Shpilka and Yehudayoff~\cite{DSY09} which shows that
up to  %
some minor technical conditions (which are not very hard to satisfy), factors of a polynomial $f \in \F[X_1, X_2, \ldots, X_N, Y]$ of the form $Y-p(X_1, X_2, \ldots, X_N)$ where $p \in \F[X_1, X_2, \ldots, X_N]$ can be expressed as polynomials in the coefficients when viewing $f$ as an element of $\F[X_1, X_2, \ldots, X_N][Y]$. This is relevant since if a
set of $t$
polynomials
is  %
algebraically dependent, then
there is a non-zero $t$-variate polynomial which vanishes when composed with this tuple. We use this \emph{vanishing} to prove the lemma.

The PIT results follows a similar initial setup and use of  \expref{Lemma}{lem:using algebraic dependence-intro}. We then use a result of Forbes~\cite{Forbes-personal} to show that the polynomial computed by $C$ has a  monomial of small support, which is then detected using the standard idea of using Shpilka-Volkovich generators~\cite{SV09}. 

\subsection{Organization of the paper}
The rest of the paper is organized
as follows.
In \expref{Section}{sec:prelims}, we state some preliminary definitions and results that are used elsewhere in the paper. In \expref{Section}{sec: alg dep}, we describe our use of low algebraic rank and prove \expref{Lemma}{lem:expressing as functions of the basis}. We prove \expref{Theorem}{thm:lower bound} in \expref{Section}{sec:lower bounds} and \expref{Theorem}{thm:PIT} in \expref{Section}{sec:PIT}. We end with some open questions in \expref{Section}{sec:open questions}.

\section{Preliminaries}~\label{sec:prelims}
In this section we
introduce  %
some
notation  %
and definitions for the rest of the paper.
\subsection{Notation}\label{sec:notation}
\begin{enumerate}
\item For an integer $i$, we denote the set $\{1, 2, \ldots, i\}$ by $[i]$.
\item By $\overline{X}$, we mean the set
$\{X_1, X_2, \ldots, X_N\}$ of variables.  %
\item For a field $\F$, we use $\F[\overline{X}]$ to denote the ring of all polynomials in $X_1, X_2, \ldots, X_N$ over the field $\F$. For brevity, we denote a polynomial $P(X_1, X_2, \ldots, X_N) \in \F[\overline{X}]$ by $P(\overline{X})$.
\item The support of a monomial $\alpha$ is the set of variables which appear with a non-zero exponent in $\alpha$. %
\item  We say that a function $f(N)$ is quasipolynomially bounded in $N$ if there exists a positive absolute constant $c$, such that  for all $N$ sufficiently large, $f(N) < \exp(\log^c N)$. For brevity, if $f$ is quasipolynomially bounded in $N$, we say that $f$ is quasipolynomial in $N$. 
\item In this paper, unless otherwise stated, $\F$ is a field of characteristic zero. 
\item Given a polynomial $P$ and a valid monomial ordering $\Pi$, the leading monomial of $P$ is the monomial with a nonzero coefficient in $P$ which is maximal according to $\Pi$. Similarly, the trailing monomial in $P$ is the monomial which is minimal among all monomials in $P$ according to $\Pi$.
\item All our logarithms are to the base $\eee$.
\end{enumerate}

\subsection{Algebraic independence}
We formally defined the notion of algebraic independence and algebraic rank in \expref{Definition}{def:alg-indepence}. For more on algebraic independence and related discussions, we refer the reader to the excellent survey by Chen, Kayal and Wigderson~\cite{CKW11} and earlier papers~\cite{BMS11, ASSS12}. %

For a tuple ${\cal Q} = (Q_1, Q_2, \ldots, Q_t)$ of algebraically dependent polynomials, we know that there is a nonzero
$t$-variate %
polynomial $R$ (called a $\cal Q$-annihilating polynomial) such that 
$R(Q_1, Q_2, \ldots, Q_t)$ is identically zero. A natural question is to ask, what kind of bounds on the degree of $R$ can we show, in terms of the degrees of $Q_i$. The following lemma of Kayal~\cite{Kayal09}
gives  %
an upper bound on the degree of annihilating polynomials of a set of degree-$d$ polynomials. The bound is useful to us in our proof. 
 
\begin{lemma}[Kayal~\cite{Kayal09}]~\label{lem:degree upper bound for annihilating poly}
Let $\F$ be a field and let ${\cal Q} = \{Q_1, Q_2, \ldots, Q_t\}$ be a set of polynomials of degree-$d$ in $N$ variables over the field $\F$ having algebraic rank $k$. Then there exists a $\cal Q$-annihilating polynomial
of degree at most $(k+1)\cdot d^k$.
\end{lemma}

\subsection{Complexity of homogeneous components}      
We start by defining the homogeneous components of a polynomial. 
\begin{definition}\label{def:homog components}
 For a polynomial $P$ and a positive integer $i$, we represent by $\h^i[P]$, the homogeneous component of $P$ of degree equal to $i$. 
By extension, we define $\h^{\leq i}[P]$ and $\h^{\geq i}[P]$ as follows. 
\begin{align*}
\h^{\leq i}[P] &\equiv \sum_{j = 0}^i \h^{j}[P]\,.\\
\h^{\geq i}[P] &\equiv \sum_{j = i}^{\deg(P)} \h^{j}[P]\,.
\end{align*}
We define $\h[P]$ as the ordered tuple of homogeneous components of $P$, \ie, 
\[\h[P] \equiv \left(\h^{d}[P], \h^{d-1}[P], \ldots, \h^{0}[P]\right)\, ,\] where $d$ is the degree of $P$. %
\end{definition}
We will use the following simple lemma
whose proof is fairly standard using interpolation,
and can be found in the paper~\cite{KS15-lowarity},
for instance. We sketch the proof here for completeness. 

\begin{lemma}~\label{lem:interpolation nonhomogeneous}
Let $\F$ be a field of characteristic zero, and let $P \in \F[X_1, X_2, \ldots, X_N]$ be a polynomial of degree at most $d$, in $N$ variables, such that
$P$ can be represented as 
 $$P = C(Q_{1}, Q_{2}, \ldots, Q_{t}) \, ,$$
where for every $j\in [t]$, $Q_{j}$ is a polynomial in $N$ variables, and $C$ is an arbitrary polynomial in $t$ variables. Then, there exist polynomials 
$\{Q'_{ij} : i \in [d+1], j \in [t]\}$, and for every $\ell$ such that $0\leq \ell \leq d$, there exist polynomials $C'_{\ell,1}, C'_{\ell,2}, \ldots, C'_{\ell,d+1}$ satisfying  
 $$\h^{\ell}[P] = \sum_{i = 1}^{(d+1)}  C'_{\ell,i}(Q'_{i1}, Q'_{i2}, \ldots, Q'_{it})\,.$$
Moreover, 
\begin{itemize}

\item if each of the polynomials in the set $\{Q_{j} :  j \in [t]\}$ is of degree at most $\Delta$, then every polynomial in the set  $\{Q'_{ij} : i \in [d+1], j \in [t]\}$ is also of degree at most $\Delta$;

\item  if the algebraic rank of the set %
$\{Q_{j} :  j \in [t]\}$
of polynomials %
is at most $k$, then for every $i \in [d+1]$, the algebraic rank
of the set $\{Q_{ij}' : j \in [t]\}$ of polynomials %
is also at most $k$.
\end{itemize}
\end{lemma}
\begin{proof}
The key idea is to start from $P \in \F[\overline{X}]$ and obtain a new polynomial $P' \in \F[\overline{X}][Z]$ such that for every $\ell$ such that $0 \leq \ell \leq d$, the coefficient of $Z^{\ell}$ in $P'$ equals $\h^{\ell}[P]$. Here, $Z$ is a new variable. Such a $P'$ is obtained by replacing every occurrence of the variable $X_j$ (for each $j \in [N]$) in $P$ by $Z\cdot X_j$. It is not hard to verify that such a $P'$ has the stated property. We now view $P'$
as a univariate polynomial in $Z$ with the coefficients coming from $\F(\overline{X})$. Notice that the degree of $P'$ in $Z$ is at most $d$. 
So, to recover the coefficients of a univariate polynomial of degree at most $d$, we can evaluate $P'$ at $d+1$ distinct values of $Z$ over $\F(\overline{X})$ and take an $\F(\overline{X})$ linear combination. In fact, if the field $\F$ is large enough, we can assume that all these distinct values of $Z$ lie in the base field $\F$ and we only take an $\F$ linear combination.  The properties in the  ``moreover'' part of the lemma immediately follow from this construction, and we skip the details. 
\end{proof}

\subsection{Roots of  polynomials}
We will crucially use the following result of Dvir, Shpilka, Yehudayoff~\cite{DSY09}.

\begin{lemma}[Lemma 3.1 in Dvir, Shpilka, Yehudayoff~\cite{DSY09}]~\label{lem:DSY main}
For a field $\F$, let $P \in \F[X_1, X_2, \ldots, X_N, Y ]$ be a non-zero polynomial of degree at most $k$ in $Y$. Let $f \in \F[X_1, X_2, \ldots, X_N]$ be a polynomial such that $P(X_1, X_2, \ldots, X_N, f) = 0$ and $\frac{\partial P}{\partial Y} (0, 0, \ldots, 0, f(0, 0, \ldots, 0))\neq 0$. Let 
$$P = \sum_{i = 0}^k C_i(X_1, X_2, \ldots, X_N)\cdot Y^i\,.$$ Then, for every $t \geq 0$, there exists a polynomial $R_t \in \F[Z_1, Z_2, \ldots, Z_{k+1} ]$ of degree at most $t$ such that 
\begin{equation}
\h^{\leq t}[f(X_1, X_2, \ldots, X_N)] = \h^{\leq t}[R_t(C_0, C_1, \ldots, C_k)]\,.
\end{equation}
\end{lemma} 
We also use the following standard result about zeroes of polynomials. 
\begin{lemma}[Schwartz, Zippel, DeMillo, Lipton~\cite{DL78}]~\label{lem: comb nulls}
Let $P$ be a non-zero polynomial of degree-$d$ in $N$ variables over a field $\F$. Let $S$ be an arbitrary subset of $\F$, and let $x_1, x_2, \ldots, x_N$ be random elements from $S$ chosen independently and uniformly at random. Then $$\pr[P(x_1, x_2, \ldots, x_N) = 0] \leq \frac{d}{|S|}\,.$$
\end{lemma}

The following corollary easily follows from the lemma above. 
\begin{corollary}~\label{cor: schwartz zippel many}
Let $P_1, P_2, \ldots, P_t$ be non-zero polynomials of degree-$d$ in $N$ variables over a field $\F$. Let $S$ be an arbitrary subset of $\F$ of size at least $2td$, and let $x_1, x_2, \ldots, x_N$ be random elements from $S$ chosen independently and uniformly at random. Then $$\pr[\forall i \in [t], P_i(x_1, x_2, \ldots, x_N) \neq 0] \geq \frac{1}{2}\,.$$ 
\end{corollary}

\subsection{Approximations}
We will use the following lemma of Saptharishi~\cite{Saptharishi-survey1} for numerical approximations in our calculations. 

\begin{lemma}[Saptharishi~\cite{Saptharishi-survey1}]~\label{lem:approx-new}
Let $n$ and $\ell$ be parameters such that $\ell = ({n}/2)(1-\epsilon)$ for some $\epsilon = o(1)$. For any $a, b$ such that $a, b = O(\sqrt{n})$, 
\[
\binom{n-a}{\ell - b} = \binom{n}{\ell}\cdot 2^{-a} \cdot (1 + \epsilon)^{a-2b} \cdot \exp(O(b\cdot \epsilon^2))\,.
\]
\end{lemma}

\section{Utilizing low algebraic rank}~\label{sec: alg dep}
Let ${\cal Q} = \{Q_1, Q_2, \ldots, Q_t\}$ be a set of polynomials in $N$ variables and degree at most $d$ such that the algebraic rank of ${\cal Q}$ equals $k$. Without loss of generality, let us assume that ${\cal B} = \{Q_1, Q_2, \ldots, Q_k\}$ are an algebraically independent subset of $\cal C$ of maximal size. We now show that, in some sense, this implies that all the polynomials in ${\cal Q}$ can be represented as functions of  polynomials in the set ${\cal B}$. We make this notion formal in the lemma below, which is a restatement of \expref{Lemma}{lem:using algebraic dependence-intro}. 
\begin{lemma}[\expref{Lemma}{lem:using algebraic dependence-intro} restated]~\label{lem:using algebraic dependence}
Let $\F$ be any  field of characteristic zero or sufficiently large. Let ${\cal Q} = \{Q_1, Q_2, \ldots, Q_t\}$ be a set of polynomials in $N$ variables such that the algebraic rank of ${\cal Q}$ equals $k$. Let $d_i=\deg(Q_i)$ ($i\in [t]$) and let  ${\cal B} = \{Q_1, Q_2, \ldots, Q_k\}$ be a maximal algebraically independent subset of ${\cal Q}$. Then, there exists an $\overline{a} = (a_1, a_2, \ldots, a_N)$ in $\F^N$ and polynomials $F_{k+1}, F_{k+2}, \ldots, F_{t}$  in $k$ variables such that $\forall i \in \{k+1, k+2, \ldots, t\}$
$$Q_i(\overline{X} + \overline{a}) = \h^{\leq d_i}\left[F_i(Q_1(\overline{X} + \overline{a}), Q_2(\overline{X} + \overline{a}), \ldots, Q_k(\overline{X} + \overline{a})) \right]\,.$$
\end{lemma}

\begin{proof}
Let $d$ be defined as $\max_i\{d_i\}$. 
Let us consider any $i$ such that $i \in \{k+1, k+2, \ldots, t\}$. {}From the statement of the lemma, 
it follows that the set of polynomials in the set ${\cal B} \cup \{Q_i\}$ are algebraically dependent. Therefore, there exists a nonzero polynomial $A_i$ in $k+1$ variables such that $A_i(Q_1, Q_2, \ldots, Q_k, Q_i) \equiv 0$. Without loss of generality, we choose such a polynomial with the smallest total degree. {}From the upper bound on the degree of the annihilating polynomial from \expref{Lemma}{lem:degree upper bound for annihilating poly}, we can assume that the degree of $A_i$ is at most $(k+1)d^k$. Consider the polynomial $A'_i(\overline{X}, Y)$ defined by $$A'_i(\overline{X}, Y) = A_i(Q_1(\overline{X}), Q_2(\overline{X}), \ldots, Q_k(\overline{X}), Y)\,.$$
We have the following observation about properties of $A'_i$.
\begin{observation}~\label{obs:internal1}
$A'_i$ satisfies the following
conditions.
\begin{itemize}
\item $A'_i$ is not identically zero.
\item The $Y$ degree of $A'_i$ is at least one. 
\item $Q_i(\overline{X})$ is a root of the polynomial $A'_i$, when viewing it as a polynomial in the $Y$ variable with coefficients coming from $\F(\overline{X})$.
\end{itemize}
\end{observation}
\begin{proof}
We prove the items in
sequence.  %
\begin{itemize}
\item If $A'_i$ is identically zero, then it follows that $Q_1, Q_2, \ldots, Q_k$ are algebraically dependent, which is a contradiction. 
\item If $A'_i(\overline{X} , Y)$ does not depend on the variable $Y$, then by definition, it follows that $A_i(Q_1, Q_2, \ldots, Q_k, Y)$ does not depend on $Y$. Hence, $A_i(Q_1, Q_2, \ldots, Q_k, Q_{i})$ does not depend on $Q_{i}$ but is identically zero. This contradicts the algebraic independence of $Q_1, Q_2, \ldots, Q_k$.
\item This item follows from the fact that the polynomial obtained by substituting $Y$ by $Q_i$ in $A'_i$ equals $A_i(Q_1, Q_2, \ldots, Q_k, Q_i)$, which is identically zero. \qedhere
\end{itemize}
\end{proof}

Our aim now is to invoke \expref{Lemma}{lem:DSY main} for the polynomial $A'_i$, but first, we need to verify that the conditions in the hypothesis of \expref{Lemma}{lem:DSY main} are satisfied. Let the polynomial $A''_i$ be defined as the first order derivative of $A'_i$ with respect to $Y$. Formally, 
$$A''_i = \frac{\partial{A'_i}}{\partial Y}\,.$$
We proceed with the following claim, the proof of which we defer to the end. 
\begin{claim}~\label{clm:derivative nonzero}
The polynomial $A''_i$ is not an identically zero polynomial and ${A''_i|}_{Y = Q_i}$ is not identically zero. 
\end{claim}
For the ease of notation, we define 
\[
{L_i}(\overline{X}) = {A''_i|}_{Y = Q_i}\,.
\]
Observe that $L_i$ is a polynomial in the variables $\overline{X}$ which is not identically zero and is of degree at most $(k+1)d^{k+1}$. Let $H$ be a subset of $\F$ of size $2t(k+1)d^{k+1}$.  Then, for a uniformly random point $\overline{a}_i$ picked from $H^N$, the probability that $L_i$ vanishes at $\overline{a}_i$ is at most $1/2t$. We call the set of all points $\overline{a}_i\in H^N$ where $L_i$  vanishes as bad. Then, with a probability at least $1-1/2t$, a uniformly random element of $H^N$ is not bad. Let $\overline{a}_i \in \F^N$ be a ``not bad'' element. 
We can replace $X_j$ by $X_j + \gamma$, where $\gamma$ is the $j^{th}$ coordinate of $a_i$ and then for the resulting polynomial $L_i(\overline{X} + \overline{a}_i)$, the point $(0, 0, \ldots, 0)$ is not bad. 

We are now ready to apply \expref{Lemma}{lem:DSY main}. Let 
$$A'_i(\overline{X} , Y) = \sum_{j = 0}^{(k+1)d^k} C_j(\overline{X})\cdot Y^j\,.$$ 
Here, for every $j$, $C_j(\overline{X}) = C_j'\left(Q_1(\overline{X}),Q_2(\overline{X}), \ldots, Q_k(\overline{X} )\right)$ is a polynomial in the $\overline{X}$ variables and is the coefficient of $Y^j$ in  $A'_i(\overline{X}, Y)$  when viewed as an element of $\F[\overline{X}][Y]$. 
{}From the discussion above, we know that the following are true. 
\begin{enumerate}
\item The polynomial $A'_i(\overline{X} + \overline{a}_i, Q_i(\overline{X} + \overline{a}_i))$ is identically zero. 
\item The first derivative of $A'_i(\overline{X} + \overline{a}_i, Y)$ with respect to $Y$ does not vanish at  $(0, 0, \ldots, 0, Q_i(0, 0, \ldots, 0))$. 
\end{enumerate}
Therefore, by \expref{Lemma}{lem:DSY main}, it follows that there is a polynomial $G_i$ such that 
$$Q_i(\overline{X} + \overline{a}_i) = \h^{\leq d_i}\left[G_i(C_0(\overline{X} + \overline{a}_i), C_1(\overline{X} + \overline{a}_i), \ldots, C_{(k+1)d^k}(\overline{X} + \overline{a}_i)) \right]\,.$$
We also know that for every $j \in \{0, 1, \ldots, (k+1)d^k\}$, $C_j(\overline{X} + \overline{a}_i)$ is a polynomial in the polynomials $Q_1(\overline{X} + \overline{a}_i),Q_2(\overline{X} + \overline{a}_i), \ldots, Q_k(\overline{X} + \overline{a}_i)$. In other words, 
$$Q_i(\overline{X} + \overline{a}_i) = \h^{\leq d_i}\left[F_i(Q_1(\overline{X} + \overline{a}_i), Q_2(\overline{X} + \overline{a}_i), \ldots, Q_{k}(\overline{X} + \overline{a}_i)) \right] $$
for  a polynomial $F_i$. 

In order to prove the lemma for all values of $i \in \{k+1, k+2, \ldots, t\}$, we observe that we can pick a single value of the translation $\overline{a}$, which works for every $i\in \{k+1, k+2, \ldots, t\}$. Such an $\overline{a}$ exists because the probability that a uniformly random $p \in H^N$ is bad for some $i$ is at most $t\cdot 1/2t = 1/2$ and the translation corresponding to any such element $\overline{a}$ in $H^N$ which is not bad for every $i$ will work. The statement of the lemma then immediately follows. 
\end{proof}

We now prove \expref{Claim}{clm:derivative nonzero}.
\begin{proof}[Proof of \expref{Claim}{clm:derivative nonzero}]
We observed from the second item in \expref{Observation}{obs:internal1} that the degree of $Y$ in $A'_i$ is at least $1$. Hence, $A''_i$ is not identically zero.  
If $A''_i|_{Y = Q_i}$ is identically zero, then it follows that $\{Q_1, Q_2, \ldots, Q_k, Q_i\}$ have an annihilating polynomial of degree smaller than the degree of $A_i$, which is a contradiction to the choice of $A_i$, as a minimum degree annihilating polynomial. 
\end{proof}

\expref{Lemma}{lem:using algebraic dependence} lets us express all polynomials in a set of polynomials as a function of the polynomials in the transcendence basis. However, the functional form obtained is slightly cumbersome for us to use in our applications. We now derive the following corollary, which is easier to use in our applications. 

\begin{corollary}~\label{cor:using algebraic independence new}
Let $\F$ be any  field of characteristic zero or sufficiently large. Let ${\cal Q} = \{Q_1, Q_2, \ldots, Q_t\}$ be a set of polynomials in $N$ variables such that the for every $i \in [t]$, the degree of  $Q_i$ is equal to $d_i < d$ and  the algebraic rank of ${\cal Q}$ equals $k$. Let  ${\cal B} = \{Q_1, Q_2, \ldots, Q_k\}$ be a maximal algebraically independent subset of ${\cal Q}$. Then, there exists an $\overline{a} = (a_1, a_2, \ldots, a_N)$ in $\F^N$ and polynomials $F_{k+1}, F_{k+2}, \ldots, F_{t}$  in at most $k(d+1)$ variables such that $\forall i \in \{k+1, k+2, \ldots, t\}$
$$Q_i(\overline{X} + \overline{a}) = F_i(\h[Q_1(\overline{X} + \overline{a})], \h[Q_2(\overline{X} + \overline{a})], \ldots, \h[Q_k(\overline{X} + \overline{a})])\,.$$
\end{corollary}
\begin{proof}
Let $i$ be such that $ i \in \{k+1, k+2, \ldots, t\}$. {}From \expref{Lemma}{lem:using algebraic dependence}, we know that there exists an $\overline{a} \in \F^N$ and a polynomial $W_{i}$ such that 
\begin{equation}
Q_{i}(\overline{X} + \overline{a}) = \h^{\leq d_{i}}\left[W_{i}(Q_{1}(\overline{X} + \overline{a}), Q_{2}(\overline{X} + \overline{a}), \ldots, Q_{k}(\overline{X} + \overline{a})) \right]\,.
\end{equation}
We will now show that $\h^{\leq d_{i}}\left[W_{i}(Q_{1}(\overline{X} + \overline{a}), Q_{2}(\overline{X} + \overline{a}), \ldots, Q_{k}(\overline{X} + \overline{a})) \right]$ is actually a polynomial in the homogeneous components of the various $Q_{j}(\overline{X} + \overline{a})$ by the following procedure, which is essentially univariate polynomial interpolation. 
\begin{itemize}
\item Let $R(\overline{X}) = W_{i}(Q_{1}(\overline{X} + \overline{a}), Q_{2}(\overline{X} + \overline{a}), \ldots, Q_{k}(\overline{X} + \overline{a}))$. We replace every variable $X_j$ in $R$ by $Z\cdot X_j$ for a new variable $Z$. We view the resulting polynomial $R'$ as an element of $\F(\overline{X})[Z]$, \ie, a univariate polynomial in $Z$ with coefficients coming from the field of rational functions in the $\overline{X}$ variables. 
\item Now, observe that for any $\ell$, the homogeneous component of degree-$\ell$ of  $R$ is precisely the coefficient of $Z^{\ell}$ in $R'$. Hence, we can evaluate $R'$ for sufficiently many distinct values of $Z$ in $\F(\overline{X})$, and then take an $\F(\overline{X})$ linear combination of these evaluations to express the homogeneous components. Moreover, since $\F$ is an infinite field, without loss of generality, we can pick the values of $Z$ to be scalars in $\F$, and in this case, we will just be taking an $\F$ linear combination. 
\end{itemize}
The catch here is that after replacing $X_j$ by $Z\cdot X_j$ and substituting different values of $Z \in \F$, the polynomials $Q_{i'}(\overline{X} + \overline{a})$ could possibly lead to distinct polynomials. In general, this is bad, since our goal is to show that every polynomial in a set of algebraically dependent polynomials in a function of \emph{few} polynomials. However, the following observation comes to our rescue. Let $P$ be any polynomial in $\F[\overline{X}]$ of degree-$\Delta$ and let $P'$ be the polynomial obtained from $P$ by replacing $X_j$ by $Z\cdot X_j$. Then, 
\begin{equation}
P'(\overline{X} + \overline{a}) = \sum_{\ell = 0}^{\Delta} Z^{\ell}\cdot\h^{\ell}[P(\overline{X} + \overline{a})] \,.
\end{equation}
In particular, the set of polynomials obtained from $P'$ for different values of $Z$ are all in the linear span of homogeneous components of $P$. 

Therefore, any homogeneous component of $R$ can be expressed as a function of the set  
 \[\bigcup_{i = 1}^k \h\left[Q_i(\overline{X} + \overline{a}) \right]\] 
of polynomials.
This completes the proof of the corollary. 
\end{proof}

We now prove the following lemma, which will be directly useful in the our applications to polynomial identity testing and lower bounds in the following sections.  
\begin{lemma}~\label{lem:expressing as functions of the basis}
Let $\F$ be any  field of characteristic zero or sufficiently large. Let $P \in \F[\overline{X}]$ be a polynomial in $N$ variables, of degree equal to $n$, such that  $P$ can be represented as $$P = \sum_{i = 1}^T  F_{i}(Q_{i1}, Q_{i2}, \ldots, Q_{it})    $$
and such that the following are true.  %
\begin{itemize}
\item For each $i \in[T]$, $F_i$ is a polynomial in $t$ variables. 
\item For each $i \in [T]$ and $j \in [t]$, $Q_{ij}$ is a polynomial in $N$ variables of degree at most $d$. 
\item For each $i \in [T]$, the algebraic rank of the set 
$\{Q_{ij} : j \in [t]\}$ 
of polynomials  %
is at most $k$ and  ${\cal B}_i = \{Q_{i1}, Q_{i2}, \ldots, Q_{ik}\}$ is a maximal algebraically independent subset of $\{Q_{ij} : j \in [t]\}$. 
\end{itemize}
Then, there exists an $\overline{a} \in \F^{N}$ and polynomials $F_i'$ in at most $k(d+1)$ variables such that 
\[
P(\overline{X} + \overline{a}) = \sum_{i = 1}^T  F_i'(\h[Q_{i1}(\overline{X} + \overline{a})], \h[Q_{i2}(\overline{X} + \overline{a})], \ldots, \h[Q_{ik}(\overline{X} + \overline{a})])\,.
\]
\end{lemma}

\begin{proof}
The proof would essentially follow from the application of \expref{Corollary}{cor:using algebraic independence new} to each of the summands on the right hand side. The only catch is that the translations $\overline{a}$ could be different for each one of them. Since we are working over infinite fields, without loss of generality, we can assume that there is a good translation $\overline{a}$ which works for all the summands. 
\end{proof}

\section{Application to lower bounds}\label{sec:lower bounds}
In this section , we prove \expref{Theorem}{thm:lower bound}. But, first we discuss the definitions of the complexity measure used in the proof, the notion of random restrictions and the family of hard polynomials that we work with. 
\subsection{Projected shifted partial derivatives}~\label{sec:shifted partials prelims}
The complexity measure that we use to prove the lower bounds in this paper is the notion of \emph{projected shifted partial derivatives} of a polynomial introduced by Kayal et al.\ in~\cite{KLSS14} and subsequently used in a number of following papers~\cite{KS-full, KayalSaha14, KS15-lowarity}. 

For a polynomial $P$ and a monomial $\gamma$,  ${\frac{\partial P}{\partial \gamma}}$ is the partial derivative of $P$ with respect to $\gamma$ and for a set of monomials ${\cal M}$,  $\partial_{\cal M} (P)$ is the set of partial derivatives of $P$ with respect to monomials in ${\cal M}$. The space of $({\cal M}, m)\mhyphen$projected shifted partial derivatives of a polynomial $P$ is defined below. 
\begin{definition}[$({\cal M}, m)\mhyphen$projected shifted partial derivatives]\label{def:shiftedderivative}
For an
$N$-variate %
polynomial
\[
  P \in {\field{F}}[X_1, X_2, \ldots, X_{N}]\,,
\]
set of monomials ${\cal M}$ and a positive integer $m\geq 0$, the space of $({\cal M}, m)$-projected shifted partial derivatives of $P$ is defined as
\begin{align}
 \langle \partial_{\cal M} (P)\rangle_{m} \stackrel{\operatorname{def}}{=} \field{F}\mhyphen\operatorname{span}\left\{\mult\left[\prod_{i\in S}{X_i}\cdot g\right]  :  g \in \partial_{\cal M} (P), S\subseteq [N], |S| = m\right\}\,.
\end{align}
\end{definition}
Here, $\mult[P]$ of a polynomial $P$ is the projection of $P$ on the multilinear monomials in its support. We use the dimension of projected shifted partial derivative space of $P$ with respect to some set of monomials ${\cal M}$ and a parameter $m$ as a measure of the complexity of a polynomial. Formally, 
$$\Phi_{{\cal M}, m} (P) = \dim( \langle \partial_{\cal M} (P)\rangle_{m})\,.$$
{}From the definitions, it is straightforward to see that the measure is subadditive.
\begin{lemma}[Subadditivity]~\label{lem:subadditive}
Let  $P$ and $Q$ be any two multivariate polynomials in $\F[X_1, X_2, \ldots, X_{N}]$.  Let ${\cal M}$ be any set of monomials and $m$ be any positive integer. Then, for all scalars $\alpha$ and $\beta$
$$\Phi_{{\cal M}, m} (\alpha\cdot P + \beta\cdot Q) \leq \Phi_{{\cal M}, m} (P) + \Phi_{{\cal M}, m} (Q)\,.$$
\end{lemma} 
In the proof of \expref{Theorem}{thm:lower bound}, we need to upper bound the dimension of the span of projected shifted partial derivatives of the homogeneous component of a fixed degree of polynomials. The following lemma comes to our rescue there. 

\begin{lemma}~\label{lem:measure of homogeneous components}
Let $P$ be a polynomial of degree at most $d$. Then for every $0 \leq i \leq d$, and for every choice of parameters $m, r$ and a set $\cal M$ of monomials of degree equal to $r$, the following inequality is true.
$$\phi_{{\cal M}, m}(P) \geq  \phi_{{\cal M}, m}(\h^i[P])\,.$$
\end{lemma}
\begin{proof}
Since $\cal M$ is a subset of monomials of degree equal to $r$, all the partials derivatives are shifted by monomials of degree equal to $m$ and the operation $\mult[]$ either sets a monomial to zero or leaves it unchanged, it follows that the span of projected shifted partial derivatives of $\h^i[P]$ coincides with the span of the homogeneous components of degree-$(i-r)m$ in the space of span of projected shifted partial derivatives of $P$ itself.
The lemma then follows from the fact that dimension of a linear space of polynomials is at least as large as the dimension of the space obtained by restricting all polynomials to some fixed homogeneous component.
\end{proof}
In the next lemma, we prove an upper bound on the polynomials which are obtained by a composition of low arity polynomials with polynomials of small support. Gupta et al.~\cite{GKKS12} first proved such a bound for homogeneous depth-$4$ circuit with bounded bottom fan-in.

\begin{lemma}~\label{lem:proj shifted partials upper bound for functions of polynomials of low support}
Let $s$ be a parameter and  $Q_1, Q_2, \ldots, Q_t$ be  polynomials in $\F[\overline{X}]$ such that for every $i \in [t]$,  the support of every monomial in $Q_i$ is of size at most $s$. 
Then, for every polynomial $F$ in $t$ variables, every choice of parameters $r, m$ such that $m + rs \leq N/2$, and every set $\cal M$ of monomials of degree equal to $r$, 
$$\Phi_{{\cal M}, m}(F(Q_1, Q_2, \ldots, Q_t)) \leq N\cdot \binom{t + r}{r}\cdot \binom{N}{m + rs}\,.$$ 
\end{lemma}
\begin{proof}
By the chain rule for partial derivatives, every derivative of order $r$ of $F(Q_1, Q_2, \ldots, Q_t)$ can be written as a linear combination of products of the form 
$$\left(\frac{\partial F(Y_1, Y_2, \ldots, Y_t)}{\partial \beta_0}|_{Y_i = Q_i}\right)\cdot \prod_{1 \leq j \leq r'}\frac{\partial P_j}{\partial \beta_j} $$
where
\begin{enumerate}
\item $r'$ is at most $r$,
\item $\beta_0$ is a monomial in  variables $Y_1, Y_2, \ldots, Y_t$ of degree at most $r$,
\item for every $1 \leq j \leq r$, the polynomial $P_j$ is an element of $\{Q_1, Q_2, \ldots, Q_t\}$, and
\item for every $1 \leq j \leq r$, $\beta_j$ is a monomial in variables $X_1, X_2, \ldots, X_N$.
\end{enumerate}
Since every monomial in each $Q_i$ is of support at most $s$, every monomial in each of the products
\[
  \prod_{1 \leq j \leq r}\frac{\partial P_j}{\partial \beta_j}
\]
is of support at most $rs$. Therefore, for shifts of degree- $m$, the projected shifted partial derivatives of $F(Q_1, Q_2, \ldots, Q_t)$ (with respect to monomials in $\cal M$ which are of degree-$r$) are in the linear span of polynomials of the form
$$\mult\left[\left(\frac{\partial F(Y_1, Y_2, \ldots, Y_t)}{\partial \beta_0}|_{Y_i = Q_i}\right)\cdot \alpha\right] $$
where $\alpha$ is a multilinear monomial\footnote{If $\alpha$ is not multilinear, the term is set to zero.} of degree at most $m + rs$. 
Therefore, the dimension of this space is upper bounded by the number of possible choices of $\beta_0$ and $\alpha$. Hence
\[
  \Phi_{{\cal M}, m}(F(Q_1, Q_2, \ldots, Q_t)) \leq N\cdot \binom{t + r}{r}\cdot \binom{N}{m + rs}\,.\qedhere
\]
\end{proof}

\subsection{Target polynomials for the lower bound}\label{sec:hard poly}
In this section, we define the family of polynomials for which we
prove  %
our lower bounds. The family is a variant of the Nisan-Wigderson polynomials which were introduced by Kayal et al.\ in~\cite{KSS13}, and subsequently used in many other results~\cite{KS-full, KayalSaha14, KS15-lowarity}. We start with the following definition. 

\begin{definition}[Nisan-Wigderson polynomial families]~\label{def:NW final}
Let $n,q,e$ be arbitrary parameters with $q$ being a power of a prime, and $n,e\leq q$. We identify the set $[q]$ with the field $\F_{q}$ of $q$ elements. 
Observe that  since $n \leq q$, we have that $[n] \subseteq \F_q$. The Nisan-Wigderson polynomial with parameters $n,q,e$, denoted by $\NW_{n,q,e}$ is defined as
\[
\NW_{n,q,e}(\overline{X}) = \sum_{\substack{p(t) \in \F_q[t]\\ \deg(p) < e}} X_{1,p(1)}\dots X_{n,p(n)}\,.
\]
\end{definition}
The number of variables in $\NW_{n,q,e}$ as defined above is $N = q\cdot n$. %
The lower bounds in this paper will be proved for the polynomial $\gH$ which is a variant of the polynomial $\NW_{n,q,e}$ defined as follows. 
\begin{definition}[Hard polynomials for the lower bound]
Let $\delta \in (0,1)$ be an arbitrary constant, and let $p = N^{-\delta}$.  Let $$\gamma = \frac{N}{p} = N^{1+\delta}\,.$$
The polynomial $\gG$ is defined as 
\[
\gG = \NW_{q, n, e}\left(\sum_{i = 1}^{\gamma}X_{1,1, i}, \sum_{i = 1}^{\gamma}X_{1,2,i}, \ldots, \sum_{i = 1}^{\gamma}X_{n,q,i} \right)\,.
\]
\end{definition}

For brevity, we will denote $\gG$ by $\gH$ for the rest of the discussion. 
The advantage of using this trick\footnote{This idea came up during discussions with Ramprasad Saptharishi.} of composing with linear forms is that it becomes cleaner to show that the polynomial $\gH$ is robust under random restrictions where every variable is kept alive with a probability $p$. Since $\delta$ is an absolute constant, the number of variables in $\gH$ is at most $N^{O(1)}$. We now formally define our notion of random restrictions. 

Let $\cal V$ be the set of variables in the polynomial $\gH$. We now define a distribution ${\cal D}_p$ over the subsets of ${\cal V}$. 

\paragraph{The distribution ${\cal D}_p$:} Each variable in $\cal V$ is independently kept alive with a probability $p = N^{-\delta}$. 

The random restriction procedure  samples a $V \gets \cal D$ and then keeps only the variables in $V$ alive. The remaining variables are set to $0$. We denote the restriction of the polynomial obtained by such a restriction as $\gH|_V$. Observe that a random restriction also results in a distribution over the restrictions of a circuit computing the polynomial $\gH$. We denote by $C|_V$  the restriction of a circuit $C$ obtained by setting every input gate in $C$ which is labeled by a variable outside $V$ to $0$. 

We now show that with a high probability over restrictions sampled according to ${\cal D}_p$, the projected shifted partial derivative complexity of $\gH$ remains high. We need the following lower bound on the dimension of projected shifted partial derivatives of $\NW_{n, q, e}$.  
\begin{lemma}[\cite{KS-full, KS15-depth5}]\label{lem:KS-tight-bound}
For every $n$ and $r = O(\sqrt{n})$ there exists parameters $q,e,  \epsilon$ such that $q = \Omega(n^2)$, $N = qn$ and $\epsilon = \Theta({\log(n)}/{\sqrt{n}})$ with
\begin{align*}
  q^r & \geq  (1+\epsilon)^{2(n-r)},\\
  q^{e-r} & =   \left(\frac{2}{1+\epsilon}\right)^{n-r} \cdot \poly(q)\,. 
\end{align*}
For any $\{n,q,e,r,\epsilon \}$ satisfying the above constraints, and for $m = ({N}/2)(1 - \epsilon)$,  over any field $\F$, we have
\[
\Phi(\NW_{n,q,e}) \geq \binom{N}{m + n - r} \cdot \exp(-O(\log^2 n))\,.
\]
\end{lemma} 
We will instantiate the lemma above with the following choice of
parameters.   %
\begin{itemize}
\item $\epsilon = \frac{4\log n}{\sqrt{n}}$,
\item $r = \sqrt{n}$,
\item $q = n^{10}$.
\item  %
  We will set the parameter $s$ to be equal to $\frac{\sqrt{n}}{100}$.
\end{itemize}
It is straightforward to check that for the above choice of parameters, there is a choice of $e$ such that 
\begin{align*}
  q^r & \geq  (1+\epsilon)^{2(n-r)}\,,\\
  q^{e-r} & =   \left(\frac{2}{1+\epsilon}\right)^{n-r} \cdot \poly(q)\,. 
\end{align*}
Therefore,  for $m = ({N}/{2})(1 - \epsilon)$,  over any field $\F$, we have
\[
\Phi(\NW_{n,q,e}) \geq \binom{N}{m + n - r} \cdot \exp(-O(\log^2 n))\,.
\]
We are now ready to prove our main lemma for this section. 
\begin{lemma}\label{lem: robustness under random restrictions}
 With a probability at least $1-o(1)$ over $V \leftarrow {\cal D}_p$, there exists   a subset of variables $V' \subseteq V$ such that $|V'| = N$ and  
\[
\Phi(\gH|_{V'}) \geq \binom{N}{m + n - r} \cdot \exp(-O(\log^2 n))\,.
\]
\end{lemma} 
\begin{proof}
To prove the lemma, we first show that with a high probability over the random restrictions, the polynomial $P|_{V}$ has the polynomial $\NW_{n, q, e}$ as a projection by setting some variables to zero. Combining this with \expref{Lemma}{lem:KS-tight-bound} would complete the proof. We now fill in the details. 

Let $i \in [N]$. Then, the probability that all the variables in the set $A_{i,j} = \{X_{i, j, \ell} : \ell \in [\gamma]\}$ are set to zero by the random restrictions is equal to $(1-p)^{\gamma} \leq \exp(-\Theta(N))$. Therefore, the probability that there exists an $i \in [n], j \in [q]$ such that all the variables in the set $A_{i,j}$ are set to zero by the random restrictions, is at most $N\cdot \exp(-\Theta(N)) = o(1)$. We now argue that if this event does not happen (which is the case with probability at least $1-o(1)$), then the dimension of the projected shifted partial derivatives is large. 

For every $i, j$, let $A_{i,j}'$ be the subset of $A_{i,j}$ which has not been set to zero. We know that for every $i, j$, $A_{i,j}'$ is non-empty. Now, for every $i, j$, we set all the elements of  $A_{i,j}'$ to zero except one. Observe that the polynomial obtained from $\gH$ after this restriction is exactly the polynomial $NW_{n, q, e}$
up to  %
a relabeling of variables. 
Now, from \expref{Lemma}{lem:KS-tight-bound}, our claim follows. 
\end{proof}

\subsection{Proof of \expref{Theorem}{thm:lower bound}}~\label{sec: thm lower bounds}

To
prove  %
our lower bound, we show that under a random restriction from the distribution ${\cal D}_p$, the dimension of the linear span of  projected shifted partial derivatives of any $\qspnewn$ circuit $C$ is small with a high probability if the size of the $C$ is \emph{not too large}. Comparing this with the lower bound on the dimension of projected shifted partials of the polynomial $\gH$ under random restrictions from \expref{Lemma}{lem: robustness under random restrictions}, the lower bound follows. We now proceed along this outline and prove the following lemma.
\begin{lemma}[Upper bound on complexity of circuits]~\label{lem: circuit complexity bound}
Let $m, r, s$ be parameters such that $m + rs \leq N/2$. Let ${\cal M}$ be any set of multilinear monomials of degree-$r$. Let $C $ be an arithmetic circuit computing a homogeneous polynomial of degree-$n$ such that 
$$C = \sum_{i = 1}^T  C_{i}(Q_{i1}, Q_{i2}, \ldots, Q_{it})$$  where
\begin{itemize}
\item for each $i \in[T]$, $C_i$ is a polynomial in $t$ variables, and
\item for each $i \in [T]$, the algebraic rank of the set $\{Q_{ij} : j \in [t]\}$ of polynomials is at most $k$. 
\end{itemize}
For each $i \in [T]$ and $j \in [t]$, let $S_{ij}$ be the set of monomials with nonzero coefficients in $Q_{ij}$. If 
$$\left|\bigcup_{i \in [T], j \in [t]} S_{ij} \right| \leq N^{\frac{\delta s}{2}} $$
then, with a probability at least $1-o(1)$ over ${V \leftarrow {\cal D}_p}$\footnote{This is the distribution defined in \expref{Section}{sec:hard poly}, where every variable is kept alive with a probability $N^{-\delta}$ for a constant $\delta \in (0,1)$.} for all subsets $V'$ of $V$ of size at most $N$
$$\Phi(C|_{V'}) \leq TN\binom{k(n+1) + r}{r}\binom{N}{m + rs}\,.$$
\end{lemma}
\begin{proof}
We prove the lemma by first using random restrictions to simplify the circuit into one with bounded bottom support, and then utilizing the tools tools developed in \expref{Section}{sec: alg dep} and \expref{Section}{sec:shifted partials prelims}  to conclude that the dimension of the space of projected shifted partial derivatives of the resulting circuit is small.

\paragraph{Step (1): Random restrictions.}
{}From the definition of random restrictions, every variable is kept alive independently with a probability $p = N^{-\delta}$. So, the probability that a monomial of support at least $s$ survives the restrictions is at most $N^{-\delta s}$. Therefore, by  linearity of expectations, the expected number of monomials of support at least $s$ in $\bigcup_{i \in [T], j \in [t]} S_{ij} $ which survive the random restrictions is at most 
$$\left|\bigcup_{i \in [T], j \in [t]} S_{ij} \right|\cdot N^{-\delta s} \leq N^{-\frac{\delta s}{2}}\,.$$
So, by Markov's inequality, the probability that at least one monomial of support at least $s$ in $\bigcup_{i \in [T], j \in [t]} S_{ij} $ survives the random restrictions is $o(1)$. Let $V'$ be any subset of the surviving set of variables of size $N$. For the rest of the proof, we assume that all the variables outside the set $V'$ are set to zero.
Restrictions which set all  monomials of support at least $s$ in $\bigcup_{i \in [T], j \in [t]} S_{ij}$ to zero are said to be good.  
\paragraph{Step (2): Using low algebraic rank.}
In this step, we assume that we are given a good restriction $C'$ of the circuit $C$. Let 
$$C' = \sum_{i = 1}^T  C_{i}'(Q_{i1}', Q_{i2}', \ldots, Q_{it}') $$
where for every $i \in [T], j \in [t]$, all monomials of $Q_{ij}'$ have support at most $s$. Observe that random restrictions cannot increase the algebraic rank of a set of polynomials. Therefore, for every $i \in [T]$, the algebraic rank of the set $\{Q'_{ij} : j \in [t]\}$ of polynomials is at most $k$. For ease of notation, let us assume that the algebraic rank is equal to $k$. Without loss of generality, let the set ${\cal B}_i = \{Q'_{i1}, Q'_{i2}, \ldots, Q'_{ik}\}$ be the set guaranteed by  \expref{Lemma}{lem:expressing as functions of the basis}. We know that there exists an $\overline{a}\in \F^{N}$ and polynomials $\{F_i ': i \in [T]\}$ such that 
\begin{equation}
 C'(\overline{X} + \overline{a}) = \sum_{i = 1}^T  F_{i}'(\h\left[ Q_{i1}'(\overline{X} + \overline{a})\right], \h\left[ Q_{i2}'(\overline{X} + \overline{a})\right], \ldots, \h\left[ Q_{ik}'(\overline{X} + \overline{a})\right])\,.
 \end{equation} 
Moreover, since $C(\overline{X})$ (and hence $C'(\overline{X})$) is a homogeneous polynomial of degree-$n$, the following is
true.  %
\begin{equation}\label{eqn: 1 in lower bound}
C'(\overline{X}) = \h^{n} \left[\sum_{i = 1}^T  F_{i}'(\h\left[ Q_{i1}'(\overline{X} + \overline{a})\right], \h\left[ Q_{i2}'(\overline{X} + \overline{a})\right], \ldots, \h\left[ Q_{ik}'(\overline{X} + \overline{a})\right])\right]\,.
\end{equation}
An important observation here is that for the rest of the argument, we can assume that the degree of every polynomial $Q_{ij}'(\overline{X} + \overline{a})$ is at most $n$. If not, we can simply replace any such high degree $Q_{ij}'(\overline{X} + \overline{a})$ by 
\[
  \h^{\leq n}\left[Q_{ij}'(\overline{X} + \overline{a})\right]\,.
\]
We claim that the equality~\ref{eqn: 1 in lower bound} continues to hold. 
This is because the higher degree monomials of $Q_{ij}$ do not participate in the computation of the lower degree monomials. The only monomials which could potentially change by this substitution are the ones with degree strictly larger than  $n$. 
\paragraph{Step (3): Upper bound on $\Phi_{{\cal M}, m}(C'(\overline{X}))$.}
Let $R$ be defined the polynomial 
\begin{equation}
R = \sum_{i = 1}^T  F_{i}'(\h\left[ Q_{i1}'(\overline{X} + \overline{a})\right], \h\left[ Q_{i2}'(\overline{X} + \overline{a})\right], \ldots, \h\left[ Q_{ik}'(\overline{X} + \overline{a})\right])\,.
\end{equation}
Note that if the support of every monomial in a polynomial $Q_{ij}'(\overline{X})$ is at most $s$, then for every translation $\overline{a} \in \F^N$ the support of every monomial in $Q_{ij}'(\overline{X} + \overline{a})$ is also at most $s$. {}From \expref{Lemma}{lem:proj shifted partials upper bound for functions of polynomials of low support} and from \expref{Lemma}{lem:subadditive}, it is easy to see that 

$$\Phi_{{\cal M}, m}(R) \leq TN \binom{k(n+1) + r}{r}\binom{N}{m + rs}\,.$$
{}From \expref{Lemma}{lem:measure of homogeneous components}, it follows that 
$$\Phi_{{\cal M}, m}(C'(\overline{X})) \leq \Phi_{{\cal M}, m}(R) \leq TN\binom{k(n+1) + r}{r}\binom{N}{m + rs}\,.$$
Observe that steps (2) and (3) of the proof are always successful if the restriction in step 1 is good, which happens with a probability at least $1-o(1)$. So, the lemma follows. 
\end{proof}

We now complete the proof of \expref{Theorem}{thm:lower bound}.  
\begin{proof}[Proof of \expref{Theorem}{thm:lower bound}]
If the size of the circuit $C$ is at least $N^{({\delta}/{2})\sqrt{n}}$, then we are done. Else, the size of $C$ is at most $N^{({\delta}/{2})\sqrt{n}}$.  This implies that the total number of monomials in all the polynomials $Q_{ij}$ together is at most $N^{({\delta}/{2})\sqrt{n}}$. 
{}From \expref{Lemma}{lem: circuit complexity bound} and \expref{Lemma}{lem: robustness under random restrictions}, it follows that  there exists a subset $V'$ of variables of size $N$ such that both the following inequalities are
true.  %
\begin{equation}
\Phi_{{\cal M}, m}(C|_{V'}) \leq TN\binom{k(n+1) + r}{r}\binom{N}{m + rs}
\end{equation}
and
\begin{equation}
\Phi_{{\cal M}, m}(\gH|_{V'}) \geq \binom{N}{m + n-r} \cdot \exp(-\log^2 n)\,. 
\end{equation}
Since $C$ computes $\gH$, it must be the case that 
$$T \geq \frac{\binom{N}{m + n-r} \cdot \exp(-\log^2 n)}{N\binom{k(n+1) + r}{r}\binom{N}{m + rs}}\,.$$
Plugging in the value of the parameters from \expref{Section}{sec:hard poly}, and  approximating using \expref{Lemma}{lem:approx-new}, we immediately get  
\[
\binom{N}{m + n -r} = \binom{N}{m} \cdot (1 + \epsilon)^{2(n-r)} \cdot \exp(O((n-r)\cdot \epsilon^2))
\]
and
\[
\binom{N}{m + rs} = \binom{N}{m} \cdot (1 + \epsilon)^{2rs} \cdot \exp(O(rs\cdot \epsilon^2))\,.
\]
Moreover, $\binom{k(n+1) + r}{r} \leq (enk)^r \leq \exp(2\sqrt{n}\cdot \log n)$. Taking the ratio and substituting the values of the parameters, we get 
\[
  T \geq \exp{(\Omega(\sqrt{n}\log N}))\,.\qedhere
\]
\end{proof}

\section{Application to polynomial identity testing}~\label{sec:PIT}
In this section we give an application of the ideas developed in \expref{Section}{sec: alg dep} to the question of polynomial identity testing and prove \expref{Theorem}{thm:PIT}. We start by  formally defining the notion of a hitting set. 
\paragraph{Hitting set.} Let ${\cal S}$ be a set of polynomials in $N$ variables over a field $\F$. Then, a set  ${\cal H} \subseteq \F^{N}$ is said to be a \emph{hitting set} for the class ${\cal S}$, if for every polynomial $P \in {\cal S}$ such that $P$ is not identically zero, there exists a $p \in {\cal H}$ such that $P(p) \neq 0$.  

For our PIT result, we show that any nonzero polynomial $P$ in the circuit class we consider, has a monomial of low support. A hitting set can then be constructed by the standard techniques using the Shpilka-Volkovich generator~\cite{SV09}.  

\begin{lemma}[Shpilka-Volkovich generator~\cite{SV09}\footnote{See Corollary 3.15 in~\cite{Forbes-personal}.}]~\label{lem:SV gen}
Let $\F$ be a field of characteristic zero. For every $\ell, d, N$, there exists a set ${\cal H} \subseteq \F^N$ of size at most $(O(Nd))^{\ell}$ such that for every nonzero polynomial $P$ of degree at most $d$ in $N$ variables which contains a monomial of support at most $\ell$, there exists an $h \in {\cal H}$ such that $P(h) \neq 0$. Moreover, the set ${\cal H}$ can be constructed in time  $\poly(N, d, \ell)\cdot (O(Nd))^{\ell}$. 
\end{lemma}

The following lemma is our main technical claim. 
\begin{lemma}~\label{lem:low support monomial}
Let $\F$ be a field of characteristic zero. Let $P$ be a  homogeneous polynomial of degree-$\Delta$ in $N$ variables such that  $P$ can be represented as $$P = \sum_{i = 1}^T  C_{i}(Q_{i1}, Q_{i2}, \ldots, Q_{it})    $$
such that the following are true.  %
\begin{itemize}
\item For each $i \in[T]$, $C_i$ is a polynomial in $t$ variables. 
\item For each $i \in [T]$ and $j \in [t]$, $Q_{ij}$ is a polynomial of degree at most $d$ in $N$ variables. 
\item For each $i \in [T]$, the algebraic rank of the set $\{Q_{ij} : j \in [t]\}$ of polynomials  is at most $k$. 
\end{itemize}
Then, the trailing monomial of $P$ has support at most $$2\eee^3d\cdot(\ln \left({T(\Delta + 1)}\right) + (d+1)k\ln {\left(2(d+1)k\right)} + 1)\,.$$ Here, $\eee$ is
Euler's constant. 
\end{lemma}

In order to prove \expref{Lemma}{lem:low support monomial}, we follow the outline of proving \emph{robust} lower bounds for arithmetic circuits, described and used by Forbes~\cite{Forbes-personal}. This essentially amounts to showing that the trailing monomial of $P$ has small support. We use the following result of Forbes~\cite{Forbes-personal} in a blackbox manner which greatly simplifies our proof. 

\begin{lemma}[Proposition 4.18 in Forbes~\cite{Forbes-personal}]\label{lem:Forbes small support monomial}
Let $R(\overline{X})$ be a polynomial in $\F[\overline{X}]$ such that $$R(\overline{X}) = \sum_{i=1}^{T} F_{i}(Q_{i1}, Q_{i2}, \ldots, Q_{it})$$
and for each $i\in [T]$ and $j \in [t]$, the degree of $Q_{ij}$ is at most $d$. Let $\alpha$ be the trailing monomial of $R$. Then, the support of $\alpha$ is at most  $2\eee^3d(\ln {T} + t\ln {2t} + 1)$, where $\eee$ is
Euler's constant.
\end{lemma}
We now proceed to prove \expref{Lemma}{lem:low support monomial}. 
\begin{proof}[Proof of \expref{Lemma}{lem:low support monomial}]
Recall that  our goal is to show that the polynomial $P$, which can be represented as
$$P = \sum_{i = 1}^T  C_{i}(Q_{i1}, Q_{i2}, \ldots, Q_{it})\, ,    $$
has a trailing monomial of small support. 

For every $i \in [T]$, let ${\cal Q}_i = \{Q_{i1}, Q_{i2}, \ldots, Q_{it}\}$ and let ${\cal Q}_i$ be of algebraic rank $k_i$. 
Without loss of generality, let us assume  the sets ${\cal B}_i = \{Q_{i1}, Q_{i2},\ldots, Q_{i{k_i}}\}$ are the sets guaranteed by 
 \expref{Lemma}{lem:expressing as functions of the basis}. This implies that there exist polynomials $F_1, F_2, \ldots, F_T$ and $\overline{a}\in \F^N$ such that
\begin{equation}
P(\overline{X}+\overline{a}) = \left[\sum_{i=1}^T F_{i}(\h[Q_{i1}(\overline{X} + \overline{a})], \h[Q_{i2}(\overline{X} + \overline{a})], \ldots, \h[Q_{ik_i}(\overline{X} + \overline{a})])\right]\,.
\end{equation}
Since each $k_i \leq k$, for the ease of notation, we assume that each $k_i = k$.
Observe that if $P$ is a homogeneous polynomial of degree  $ \deg(P) \leq \Delta$, then, $$\h^{\deg(P)}[P(\overline{X}+\overline{a})] \equiv P(\overline{X})\,.$$ 
So, from \expref{Lemma}{lem:interpolation nonhomogeneous}, it follows that there exist
$k(d+1)$-variate  %
polynomials $F_1', F_2', \ldots, F_{T(\Delta+1)}'$ and a set   $\{Q'_{ij}: i \in [T(\Delta+1)], j \in [k] \}$  of polynomials such that 
$$P(\overline{X}) = \sum_{i=1}^{T(\Delta+1)} F'_{i}(\h[Q'_{i1}(\overline{X} + \overline{a})], \h[Q'_{i2}(\overline{X} + \overline{a})], \ldots, \h[Q'_{ik}(\overline{X} + \overline{a})])\,.$$
Moreover, every polynomial in the set $\{Q'_{ij}: i \in [T(\Delta+1)], j \in [k] \}$ has degree at most $d$. Now,  \expref{Lemma}{lem:Forbes small support monomial} implies that the trailing monomial $\alpha$ of $P(\overline{X})$ has support at most 
\[
  2\eee^3d\cdot(\ln {\left(T(\Delta + 1)\right)} + (d+1)k\ln {\left(2(d+1)k\right)} + 1)\,.\qedhere
\]
\end{proof}
We are now ready to complete the proof of \expref{Theorem}{thm:PIT}. 
\begin{proof}[Proof of \expref{Theorem}{thm:PIT}]
{}From \expref{Definition}{def:lb-model}, it follows there could be non-homogeneous polynomials $P \in {\cal C}$. So, we cannot directly use \expref{Lemma}{lem:low support monomial} to say something about them, since the proof relies on homogeneity. But, this is not a problem, since a polynomial is identically zero if and only if all its homogeneous components are identically zero. Moreover, by applying \expref{Lemma}{lem:interpolation nonhomogeneous} to every summand feeding into the top sum gate of the circuit, we get that every homogeneous component of $P$\footnote{Only the top fan-in increases by a factor of $\Delta+1$, all other parameters in \expref{Definition}{def:lb-model} remain the same.} can also be computed by a circuit similar in structure to that of $P$ at the cost of a blow up by a factor $\Delta+1$ in the top fan-in. We can then apply \expref{Lemma}{lem:low support monomial} to each of these homogeneous components to conclude that if $P$ is not identically zero, then it contains a monomial of support at most 
$$ 2\eee^3d\cdot(\ln {\left(T(\Delta + 1)^2\right)} + (d+1)k\ln {\left(2(d+1)k\right)} + 1)\,.$$ \expref{Theorem}{thm:PIT} immediately follows by detecting the low support monomial using \expref{Lemma}{lem:low support monomial} and \expref{Lemma}{lem:SV gen}.  
\end{proof}

\section{Open questions}\label{sec:open questions}
We
conclude %
with some open
questions.  %
\begin{itemize}
\item
Prove the lower bounds in the paper for a polynomial in $\VP$. We believe this is true, but it seems that we need a strengthening of the bounds proved in~\cite{KS-full}. In particular, it needs to be shown that the lower bound for $IMM$ (Iterated matrix multiplication) continues to hold when a depth-$4$ circuit is not homogeneous but the formal degree is at most the square of the degree of the polynomial itself. 

\item It would be interesting to see if there are other applications of \expref{Lemma}{lem:using algebraic dependence-intro} to questions in complexity theory. The Jacobian characterization of algebraic independence has several very interesting  applications~\cite{ASSS12, DGW09}.
\end{itemize}

\section*{Acknowledgements}
Many thanks to Ramprasad Saptharishi for answering numerous questions regarding the results and techniques in~\cite{ASSS12}. We are also thankful to Michael Forbes for sharing a draft of his paper~\cite{Forbes-personal} with us and to anonymous reviewers for comments which helped us in improving the presentation of the paper.

\bibliographystyle{tocplain}   %
\bibliography{v013a006}

\begin{tocauthors}
\begin{tocinfo}[kumar]
Mrinal Kumar \\
Rutgers University, New Brunswick, NJ\\
 mrinalkumar08\tocat{}gmail\tocdot{}com \\
\url{https://mrinalkr.bitbucket.io/}
\end{tocinfo}

\begin{tocinfo}[saraf]
Shubhangi Saraf\\
Rutgers University, New Brunswick, NJ\\
 shubhangi\tocdot{}saraf\tocat{}gmail\tocdot{}com \\
 \url{http://sites.math.rutgers.edu/~ss1984/}
\end{tocinfo}

\end{tocauthors}

\begin{tocaboutauthors}

\begin{tocabout}[kumar]
 \textsc{Mrinal Kumar} received his \phd\ in Computer Science
in May 2017 from \href{http://www.rutgers.edu/}{Rutgers University}
where he was advised by \href{http://www.math.rutgers.edu/~sk1233/}{Swastik Kopparty} and \href{https://www.math.rutgers.edu/~ss1984/}{Shubhangi Saraf}. His research interests are in Arithmetic and Boolean circuit complexity and Error Correcting Codes. 
Mrinal spent his undergrad years at \href{https://www.iitm.ac.in/}{IIT Madras} and owes his interest in Complexity Theory to a delightful class on the topic taught by \href{http://www.cse.iitm.ac.in/~jayalal/}{Jayalal Sarma}. Apart from theory, he finds great joy in test cricket and in the adventures of Calvin \& Hobbes.   
\end{tocabout}

\begin{tocabout}[saraf]
\textsc{Shubhangi Saraf} grew up in Pune, India. She received her \phd\
in computer science from the Massachusetts Institute of Technology
in 2011 under the guidance of
\href{http://madhu.seas.harvard.edu/}{Madhu Sudan}.
Shubhangi is broadly interested in complexity theory, coding theory
and pseudorandomness. Recently she has been captivated
by questions related to understanding the power and limitations
of algebraic computation, as well as to understanding the
potential of \emph{locality} in algorithms for codes.

\noindent
Shubhangi discovered her love for mathematics in her
high school years at the 
\href{https://www.bprim.org}{Bhaskara\-charya Pratishthana},
an educational and research institute in mathematics in Pune,
under the guidance and mentoring of her teacher
Mr. Prakash Mulabagal.  Mr. Prakash ran an amazing program
aimed at getting high school students from across Pune
introduced to the  joy of math and the sciences beyond what
any school curriculum in Pune could possibly attempt to do.
Shubhangi owes a great deal of her enthusiasm for
math problem solving to Mr. Prakash, and also to being able,
through the Bhaskaracharya Pratishthana program,
to make close friends in Pune who were into the same thing.

\noindent
Thanks to this nurturing environment, Shubhangi got involved
in math competitions and represented India twice at the
International Mathematical Olympiad (IMO), once winning a
bronze medal (2002) and once a silver (2003).

\noindent
She went on to do her undergraduate studies in Mathematics
at MIT, graduating in 2007.  She did not really know that
she wanted to stay on in academia until her junior year
when she spent a year abroad as a \emph{mathmo} at
Cambridge University in the UK where she took fantastic
courses by Tim Gowers and Imre Leader. Once back at MIT,
in summer 2006, she did a research project with Igor Pak
at MIT, which gave her a lot of confidence and encouragement.
She was also fortunate to take some more great courses at MIT;
``Randomized algorithms'' by David Karger and
``Complexity theory'' by Madhu Sudan were particularly
influential. The support and encouragement from her MIT
mentors eventually got her on the path to theoretical
computer science.

\noindent
In her spare time Shubhangi enjoys reading, cooking,
long walks, and exploring caf\'es and restaurants. Her
little toddler is a constant source of joy and amazement,
and she also makes sure there isn't much time to spare.
\end{tocabout}
\end{tocaboutauthors}

\end{document}